\documentclass[a4paper,UKenglish,cleveref,autoref,thm-restate]{lipics-v2021}

\usepackage{amsmath,amssymb,hyperref,mathrsfs}

\newcommand{\ccsp}{\textup{\textsf{\#CSP}}}
\newcommand{\hol}{\textup{\textsf{Holant}}}
\newcommand{\plhol}{\textup{\textsf{Pl-Holant}}}
\newcommand{\symhol}{\textup{\textsf{SymB-Holant}}}
\newcommand{\symplhol}{\textup{\textsf{SymB-Pl-Holant}}}
\newcommand{\symplshol}{\textup{\textsf{SymB-Pl-S-Holant}}}
\newcommand{\gds}{\textup{\textsf{\#GDS}}}
\newcommand{\plgds}{\textup{\textsf{Pl-\#GDS}}}
\newcommand{\symgds}{\textup{\textsf{SymB-\#GDS}}}

\newcommand{\symplsgds}{\textup{\textsf{SymB-Pl-S-\#GDS}}}
\newcommand{\prob}[1]{\textup{\textsc{#1}}}

\newcommand{\abs}[1]{\lvert #1\rvert}
\newcommand{\ari}[1]{\text{arity}(#1)}
\newcommand{\degree}[1]{\operatorname{deg}(#1)}

\newcommand{\var}[1]{\text{Var}(#1)}

\newcommand{\sg}{\mathrm{\Omega}}
\newcommand{\T}{\mathrm{T}}

\title{The Counting General Dominating Set Framework}

\author{Jiayi Zheng}{Key Laboratory of System Software (Chinese Academy of Sciences), Beijing, China;\\ Institute of Software, Chinese Academy of Sciences, Beijing, China;\\ University of Chinese Academy of Sciences, Beijing, China}{zhengjy@ios.ac.cn}{https://orcid.org/0009-0005-5728-3616}{}
\author{Boning Meng}{Key Laboratory of System Software (Chinese Academy of Sciences), Beijing, China;\\ Institute of Software, Chinese Academy of Sciences, Beijing, China;\\ University of Chinese Academy of Sciences, Beijing, China;\\ University of Regensburg, Regensburg, Germany}{mengboning2013@gmail.com}{https://orcid.org/0009-0006-0088-1639}{}

\authorrunning{Jiayi Zheng and Boning Meng}

\Copyright{Jiayi Zheng and Boning Meng}


\funding{Jiayi Zheng is supported by National Key R\&D Program of China (2023YFA1009500), NSFC 62532014 and NSFC 62272448.}

\acknowledgements{We are deeply grateful to our PhD advisor Mingji Xia for his guidance and support throughout this work. He provided generous direction and advice on refining gadget construction in $\gds$ and on shaping the research trajectory from Theorem~\ref{thm:times} to $\hol^k$. Despite his contributions, he graciously declined to be a coauthor. We would like to thank Juqiu Wang for insightful discussions that helped us consolidate the framework of $\gds$. We would also like to thank an anonymous reviewer for pointing out an issue in the proof of Theorem~\ref{thm:mainds} in a preliminary version of this paper, as well as other anonymous reviewers for comments.\\ \indent The original manuscript was prepared with AI assistance limited to wording refinements. During revision, OpenAI's ChatGPT assisted with literature searches, exposition, checking arguments and calculations, and the proof development of Lemma~\ref{lem:plhol4}. The authors reviewed the incorporated material and take full responsibility for the paper.}

\nolinenumbers

\begin{document}

\maketitle

\begin{abstract}
We introduce a new framework of counting problems called \textsf{\#GDS} that generalizes \textsc{\#$(\sigma, \rho)$-Set}, a class of domination-type problems that includes counting dominating sets and counting total dominating sets. We explore the intricate relation between \textsf{\#GDS} and the well-known \textsf{Holant} framework. We adapt the technique of gadget construction of \textsf{Holant} to the \textsf{\#GDS} framework; using this technique, we prove the \#P-completeness of counting dominating sets for 3-regular planar bipartite simple graphs. Using interpolation for powered \textsf{Holant} problems and a special reduction method via symmetric bipartite graphs, we also prove the \#P-completeness of counting total dominating sets for the same graph class.
\end{abstract}

\section{Introduction}
\subsection{Background and Related Work}
Counting problems are studied extensively in the field of computational complexity. Whereas decision problems study whether a problem has a solution, counting problems aim to determine how many solutions a problem has. The $\hol$ framework (see Definition~\ref{def:Hol}) has been proposed to encompass a wide range of counting problems, and is a generalization of $\ccsp$ (counting constraint satisfaction problems). The $\hol$ framework can express many well-known problems such as counting perfect matchings or vertex covers of a graph. There have been extensive studies of complexity classifications of $\hol$ in the Boolean domain~\cite{shao2020real,backens2021holc,cai2022plsym,meng2025holodd}, whereas classifications over larger domains remain less complete~\cite{cai2016domgeq3,liu2025dom34,cai2025geometric,cai2026dom3}.

Powers of $\hol$ values have appeared in earlier reductions, such as those of Huang and Lu~\cite[Lemma~IV.4]{huang2012weighted} and Cai et al.~\cite[Lemma~7.4]{cai2022plsym}, in which disjoint copies are used to obtain a power of the original non-negative evaluation and thereby recover it. Backens and Goldberg~\cite{backens2020holnorm} studied approximation of the norm of $\hol$ evaluations for certain problem classes, which may have a connection to powered evaluations, since any fixed positive integer power of an evaluation determines its norm.

Many problems are also studied separately from the frameworks. One such problem concerns the \textit{dominating set}, a vertex subset such that every vertex in the graph belongs to the subset or has a neighbor in it. Unlike in $\hol$, where edges play a direct role in representing solutions, domination is defined through vertices and their neighbors, with edges only specifying adjacency. For general background on domination in graphs, see~\cite{haynes1998fundamentals,haynes1998advanced}.

The decision problem that asks if a graph $G$ has a dominating set of size at most $k$ has been studied extensively. As detailed by Johnson~\cite{johnson1984guide}, this problem was proved NP-complete for 3-regular planar graphs by Kikuno, Yoshida and Kakuda~\cite{kikuno1980cubic} and for bipartite graphs by Bertossi~\cite{bertossi1984bipartite}. Zverovich and Zverovich~\cite{zverovich1995perfect} later proved NP-completeness for planar bipartite graphs with maximum degree at most 3.

The counting version considered in this paper, denoted by \prob{\#Dominating Set}, asks for the exact number of dominating sets of a graph. Hunt et al.~\cite[Theorem~4.6~(3) \& Corollary~4.7]{hunt1998planar} proved \#P-hardness of the related problem of counting dominating sets of size at most $k$, even for planar bipartite graphs (in fact of maximum degree at most 4, as can be seen from their construction). Gaspers and Sorkin~\cite{gaspers2017alg} gave an exact exponential-time algorithm for counting dominating sets of each cardinality on graphs of maximum degree at most 3. For \prob{\#Dominating Set}, Kijima, Okamoto and Uno~\cite{kijima2011classes} gave polynomial-time algorithms for interval graphs and trapezoid graphs, and proved \#P-completeness for split graphs, chordal bipartite graphs and co-bipartite graphs; Lin gave polynomial-time algorithms for generalized series-parallel graphs~\cite{lin2019gsp}, distance-hereditary graphs~\cite{lin2021dh} and rooted path-tree bipartite graphs~\cite{lin2022subclasses}, and proved \#P-completeness for path-tree bipartite graphs~\cite{lin2022subclasses}. Bezáková et al.~\cite{bezakova2019approx} also studied approximation algorithms and hardness of approximation for graphs under certain degree restrictions.

The problem of dominating set has been generalized to a class of domination-type problems. Defined by Telle~\cite{telle1994sigmarho}, a \textit{$(\sigma, \rho)$-set} of a graph $G$, where $\sigma$ and $\rho$ are sets of natural numbers, is a set of vertices $S$ such that for every $u\in S$ ($v \notin S$, resp.), the number of vertices in $S$ adjacent to $u$ ($v$, resp.) is an element of $\sigma$ ($\rho$, resp.). A dominating set would therefore be a $(\{0, 1, \dots\}, \{1, 2, \dots\})$-set. A \textit{total dominating set} is a $(\{1, 2, \dots\}, \{1, 2, \dots\})$-set, i.e. every vertex in the graph has at least one neighbor in the set regardless of whether it is in the set itself. The decision problem of whether a graph has a total dominating set of size at most $k$ was proved NP-complete for bipartite graphs by Pfaff, Laskar and Hedetniemi~\cite{pfaff1983total}.

$(\sigma, \rho)$-sets can also represent various other problems for certain $\sigma, \rho$, such as independent dominating sets (which are dominating sets that are also independent sets, or equivalently, $(\{0\}, \{1, 2, \dots\})$-sets) and induced bounded-degree or $d$-regular subgraphs. Vadhan~\cite{vadhan2001independent} proved that counting independent dominating sets, equivalently maximal independent sets, is \#P-complete on planar bipartite graphs of maximum degree at most 5, and on regular graphs with no fixed upper bound on the degree. For the counting problems \prob{\#$(\sigma, \rho)$-Set}, there have notably been recent algorithmic and hardness results for bounded-treewidth graphs by Focke et al.~\cite{focke2023treewidthI,focke2023treewidthII}. Degree-restricted hardness has also been established by Tošić~\cite{tosic2010automata} for closely related problems such as counting fixed points of network automata.

\subsection{Overview}
Inspired by both $\hol$ and \prob{\#$(\sigma, \rho)$-Set}, we develop $\gds$, a new framework of counting problems that is a weighted generalization of the latter. It is well-suited for rigorous studies of counting domination-type problems.

A $\gds$ problem is similar in form to a $\hol$ problem, in which each vertex is assigned a signature that takes the variable assignment to the incident edges as input. In a $\gds$ problem, each vertex is also assigned a signature, but the signature takes as input the variable assignment to the vertex itself and all its neighboring vertices. The output is the sum of the weights of all variable assignments (each on all vertices), with the weight of each assignment defined as the product of the signature outputs from every vertex (see Definition~\ref{def:GDS}).

Given a \prob{\#$(\sigma, \rho)$-Set} problem, it can be expressed as a $\gds$ problem via a simple translation from $\sigma$ and $\rho$ to signatures (see Section~\ref{sec:apart}).

We find some equivalences between $\gds$ and $\hol$, but the complexity of general $\gds$ problems is not known to be covered by existing Boolean $\hol$ results.

We adapt gadget construction, a reduction technique widely used in $\hol$. This requires a redesign with a higher level of intricacy for $\gds$; for example, the assignment to vertices inside the gadget that are adjacent to external vertices must remain as input alongside the assignment to the external vertices themselves. This formulation provides a systematic way to track precise counting properties across gadget constructions with different structures. It allows us to develop the counting analysis around constructions that satisfy prescribed graph restrictions, and to use polynomial interpolation to extract the contributions required for a reduction. Combining this technique with carefully designed gadgets and methods in Galois theory, we establish hardness for \prob{\#Dominating Set} under the simultaneous restrictions of 3-regularity, planarity, bipartition and simplicity:

\begin{theorem}\label{thm:mainds}
\prob{\#3RPBS-Dominating Set},\footnote{We use \prob{$d$R-}, \prob{P-} or \prob{Pl-}, \prob{B-}, \prob{S-} to denote the restriction of the problem to $d$-regular, planar, bipartite and simple graphs respectively throughout the paper.} the problem of counting dominating sets for 3-regular planar bipartite simple\footnote{We emphasize "simple", because allowing multiple edges is somewhat inelegant for domination-type problems, as multiple edges are functionally no different than a single edge and may be perceived as a workaround in proofs.} (3RPBS) graphs, is \#P-complete.
\end{theorem}

We also study \textit{uniform} $\gds$ signatures (see Section~\ref{sec:apart}), which cover the case of $\sigma=\rho$ in \prob{\#$(\sigma, \rho)$-Set}, including \prob{\#Total Dominating Set}. On bipartite graphs without isolated vertices, when every vertex has a uniform signature, the $\gds$ value factorizes into two $\hol$ values. We investigate how algebraic properties of signatures, or symmetry between the two sides of a bipartite graph, can make these factors equal, connecting $\gds$ with "powered" versions of $\hol$ problems. We develop reductions that operate directly on powered evaluations, using a general lemma for interpolating powers of polynomials as a key tool and constructing instances with the required symmetry and structural restrictions, leading to the following theorem:

\begin{theorem}\label{thm:maintds}
\prob{\#3RPBS-Total Dominating Set}, the problem of counting total dominating sets for 3-regular planar bipartite simple (3RPBS) graphs, is \#P-complete.
\end{theorem}

For these two problems, membership in \#P follows by guessing a vertex subset and checking the required domination condition. Therefore, we only need to prove \#P-hardness.

\subsection{Organization}
In Section~\ref{sec:prelim}, we state preliminary knowledge. In Section~\ref{sec:GDS}, we formally propose the $\gds$ framework. In Section~\ref{sec:gadget}, we formulate gadget construction under $\gds$ and prove Theorem~\ref{thm:mainds}. In Section~\ref{sec:uniform}, we consider problems with specific uniform signatures and prove Theorem~\ref{thm:maintds}.

\section{Preliminaries}\label{sec:prelim}
Throughout this paper, complex numbers are assumed to be algebraic. The descriptions of the signatures occurring in an instance are included in its input size. We use $\leq_T$ and $\equiv_T$ for polynomial-time Turing reductions and equivalences; for reductions and equivalences stated for arbitrary signature sets, all signatures are given by full truth tables over a common number field fixed independently of the input instance.

\subsection{Domains and Signatures}
A \textit{domain of size $q$} refers to $\{0, 1, \dots, q-1\}$, or $\mathbb{Z}_q$ for simplicity. A \textit{domain-$q$ complex-valued signature} $f$ with $r$ variables is a function that maps elements from $\mathbb{Z}_q^r$ to $\mathbb{C}$. For any $q$-ary string $\alpha$ of length $r$, the output of $f$ on input $\alpha$ is written as $f(\alpha)$, or $f_{\alpha}$ if unambiguous. The set of variables associated with $f$ is denoted $\var{f}$, and $r$, the number of these variables, is referred to as the arity of $f$ and written as $\ari{f}$.

The \textit{Boolean domain} is $\mathbb{Z}_2$, and is the default domain throughout the paper unless specified otherwise.

For Boolean signatures, we can also represent a \textit{symmetric signature} $f$ of arity $r$ as $[f_0, f_1, \ldots, f_r]_{r}$, where each $f_i$ corresponds to the value of $f$ on all inputs with Hamming weight $i$. The subscript $r$ may be dropped when the arity is clear from context.

Let $\mathcal{EQ}$ represent the set of \textit{equality signatures}, defined by $\mathcal{EQ}=\{=_1, =_2, \ldots, =_r, \ldots\}$, where $=_r$ refers to the signature $[1, 0, \ldots, 0, 1]_r$ of arity $r$. That is, $=_r$ outputs 1 if all bits in the input are the same \textemdash\ either all 0's or all 1's \textemdash\ and 0 otherwise.

Signatures can be expressed in matrix form. For any $q\in\mathbb{Z}^+$, given a signature $f: \mathbb{Z}_q^r\to\mathbb{C}$ with ordered variables $(x_1, \ldots, x_r)$, its \textit{signature matrix} $M_{x_1\cdots x_l, x_{l+1}\cdots x_r}(f)\in\mathbb{C}^{q^l\times q^{r-l}}$ has entry $f(\alpha\beta)$ in the row indexed by $\alpha\in\mathbb{Z}_q^l$ and the column indexed by $\beta\in\mathbb{Z}_q^{r-l}$, with rows and columns ordered lexicographically. We typically choose $l=r/2$ when the arity $r$ is even, and write $M(f)$ or $M_f$ when the input grouping is clear. For instance, a Boolean binary signature $f$ is usually written as:
\[
M_{f}=\begin{pmatrix}
    f_{00} & f_{01}\\
    f_{10} & f_{11}
\end{pmatrix}.
\]
Throughout this paper, we may also directly use a matrix to refer to a signature (or the other way around) for convenience.

\subsection{The \texorpdfstring{$\hol$}{Holant} Framework}\label{sec:holant}
In this subsection, we describe the $\hol$ framework, which is detailed in \cite{cai2017complexity}.

Let $\mathcal{F}$ be a fixed set of Boolean signatures. A \textit{signature grid} over $\mathcal{F}$, denoted by $\sg(G, \pi)$, consists of a graph $G=(V, E)$ and a mapping $\pi$. $G$ may include parallel edges and self-loops; $\pi$ assigns to each vertex $v\in V$ a signature $f_v\in\mathcal{F}$ and a specified ordering of the edges incident to $v$. The arity of $f_v$ equals the degree of $v$, denoted by $\degree{v}$, and each input variable of $f_v$ corresponds to an incident edge of $v$. For any binary assignment $\sigma$ to all edges, the evaluation of the signature grid is defined as the product $\prod_{v\in V}f_v(\sigma|_{E(v)})$, where $\sigma|_{E(v)}$ denotes the ordered input tuple obtained by restricting $\sigma$ to the edges incident to vertex $v$.

\begin{definition}[$\hol$ problems~{\cite[Definition~1.1]{cai2017complexity}}]\label{def:Hol}
A $\hol$ problem, $\hol(\mathcal{F})$, parameterized by a set $\mathcal{F}$ of Boolean complex-valued signatures, is defined as follows: Given an instance that is a signature grid $\sg(G, \pi)$ over $\mathcal{F}$, the output is the partition function of $\sg$,
\[
\hol_\sg = \sum_{\sigma:E\rightarrow\{0, 1\}}\prod_{v \in V}f_v(\sigma|_{E(v)}).
\]
The bipartite $\hol$ problem $\hol(\mathcal{F}\mid\mathcal{G})$ is another form of $\hol$ problem where $G=(U, V, E)$ is a bipartite graph, with vertices in $U$ and $V$ assigned signatures from $\mathcal{F}$ and $\mathcal{G}$ respectively. We denote the left-hand side and right-hand side briefly as the LHS and RHS respectively.

A domain-$q$ $\hol$ problem, $\hol_q(\mathcal{F})$, is defined over domain-$q$ signatures, with the sum taken over all edge assignments $\sigma: E\rightarrow\mathbb{Z}_q$.
\end{definition}

We can write a single signature directly without using braces. For example, \prob{\#3R-Perfect Matching} and \prob{\#3R-Matching} can be expressed as $\hol([0, 1, 0, 0])$ and $\hol([1, 1, 0, 0])$, while \prob{\#3R-Vertex Cover} can be expressed as $\hol(=_3\mid [0, 1, 1])$, where $=_3$ is put over every vertex and $[0, 1, 1]$ is put over a new vertex added on every edge to require that at least one of the edge's incident vertices is assigned 1. For any $\mathcal{F}$, $\hol(\mathcal{F})$ is also equivalent to $\hol(=_2\mid\mathcal{F})$.

We note that multiplying a signature by a fixed non-zero scalar does not change the problem's computational complexity. This also applies to $\gds$ defined in Section~\ref{sec:GDSdef}, and $\hol^k$ defined in Section~\ref{sec:holpow}. Rescaling $m$ occurrences of a signature by a fixed non-zero scalar $c$ multiplies the $\hol$ or $\gds$ value by $c^m$, and the $\hol^k$ value by $c^{km}$. In the proofs, we may often multiply signatures by such scalars implicitly.

\subsection{Gadget Construction in \texorpdfstring{$\hol$}{Holant}}\label{app:holgadget}
In this subsection, we describe the form of gadget construction widely used as a critical reduction technique in the $\hol$ framework.

Let $\mathcal{F}$ denote a set of signatures. Similar to a signature grid, an $\mathcal{F}$-gate can be written as $\sg(G, \pi)$, but here $G=(V, E\cup D)$, where the edges consist of internal edges $E$, each of which is already connected to two vertices in $V$, and dangling edges $D$, each of which is only connected to a vertex on one side and left dangling on the other.

Let $\lvert E\rvert =n$ and $\lvert D\rvert =m$, with internal and dangling edges corresponding to $\{x_1, \dots, x_n\}$ and $\{y_1, \dots, y_m\}$ respectively. The $\mathcal{F}$-gate induces a signature $f$ on $\{0, 1\}^m$ defined by:
\[
f(y_1, \dots, y_m)=\sum_{\sigma:E\rightarrow\{0, 1\}}\prod_{v\in V}f_v\left(\hat{\sigma}|_{E(v)}\right)
\]
where $\hat{\sigma}:E\cup D\rightarrow\{0, 1\}$ extends $\sigma$ by incorporating the assignment to dangling edges.

We say a signature is realizable from $\mathcal{F}$ if it can be represented by an $\mathcal{F}$-gate. For any signature set $\mathcal{F}$, and signature $f$ which is realizable from $\mathcal{F}$, it is shown in \cite[Lemma~1.3]{cai2017complexity} that $\hol(\mathcal{F})\equiv_T\hol(\mathcal{F}\cup\{f\})$, which provides a powerful tool to build reductions.

By expressing signatures in matrix form, we find that certain constructions directly correspond to matrix operations.

Suppose $f$ and $g$ are two signatures, with $\var{f}=\{x_1, \dots, x_n\}$ and $\var{g}=\{y_1, \dots, y_m\}$. If we connect the dangling edges $\{x_{n-l+1}, \dots, x_n\}$ of $f$ to $\{y_1, \dots, y_l\}$ of $g$, we get the signature $h_1$ which satisfies:
\[
M_{h_1}=M_{x_1\dots x_{n-l}, y_{l+1}\dots y_m}=M_{x_1\dots x_{n-l}, x_{n-l+1}\dots x_n}\cdot M_{y_1\dots y_l, y_{l+1}\dots y_m}=M_f\cdot M_g.
\]
If we don't connect any edges, and just draw $f$ on top of $g$, with $f$ having $p$ edges to the left and $n-p$ edges to the right, and $g$ having $q$ edges to the left and $m-q$ edges to the right, we get the signature $h_2$ which satisfies:
\[
M_{h_2}=M_{x_1\dots x_py_1\dots y_q, x_{p+1}\dots x_ny_{q+1}\dots y_m}=M_{x_1\dots x_p, x_{p+1}\dots x_n}\otimes M_{y_1\dots y_q, y_{q+1}\dots y_m}=M_f'\otimes M_g'
\]
where $\otimes$ is the symbol for tensor product, and $M_f'$ and $M_g'$ are alternative forms of the same signatures $f$ and $g$.

\subsection{Polynomial Interpolation}\label{app:interpol}
The idea of polynomial interpolation is widely used in the field of counting complexity to retrieve the value of coefficients in a polynomial.

A basic form is that if we can obtain the value of the polynomial $\sum_{i=0}^nc_ix^i$ for $n+1$ different assignments to $x$, then we can establish a system of linear equations, where the coefficients are $x^i$ and the variables are $c_i$. Since the coefficient matrix of the system is a full-rank Vandermonde matrix, there is a unique solution for every $c_i$.

We also use a higher-dimensional form of polynomial interpolation, presented in the following lemma. Cai, Lu and Xia~\cite[Theorem~3.5]{cai2012interpolation} established the case $k=3$ in the $\hol$ setting; their interpolation argument extends to any fixed $k$ with only routine modifications.

\begin{lemma}\label{lem:interpolation}
Fix an integer $k\geq 1$, a number field $K$, a matrix $A\in K^{k\times k}$, and a vector of initial values $(x_{1, 1}, \dots, x_{1, k})^\T\in K^k$. For $s\geq 2$, define
\[
\begin{pmatrix}
    x_{s, 1}, \dots, x_{s, k}
\end{pmatrix}^\T=A\begin{pmatrix}
    x_{s-1, 1}, \dots, x_{s-1, k}
\end{pmatrix}^\T.
\]
Suppose that:
\begin{enumerate}
    \item $det(A)\neq 0$;
    \item $u(x_{1, 1}, \dots, x_{1, k})^\T\neq 0$ for every row eigenvector $u$ of $A$;
    \item For all $(d_1, \dots, d_k)\in\mathbb{Z}^k-\{(0, \dots, 0)\}$ with $d_1+\dots+d_k=0$, $\lambda_1^{d_1}\dots\lambda_k^{d_k}\neq 1$, where $\lambda_1, \dots, \lambda_k$ are distinct eigenvalues of $A$.
\end{enumerate}
For any integer $n\geq 0$, consider the homogeneous polynomial
\[
P(z_1, \dots, z_k)=\sum_{\substack{i_1, \dots, i_k\in\mathbb{Z}_{\geq 0}\\ i_1+\dots+i_k=n}}c_{i_1, \dots, i_k}z_1^{i_1}\dots z_k^{i_k}
\]
where $c_{i_1, \dots, i_k}\in K$. Then all the coefficients $c_{i_1, \dots, i_k}$ can be computed from the exact evaluations $P(x_{s, 1}, \dots, x_{s, k})$ where $1\leq s\leq\binom{n+k-1}{k-1}$, in time polynomial in $n$ and the total encoding length of these evaluations using standard exact representations in $K$.
\end{lemma}

Among the three conditions above, the third one is particularly difficult to verify. The authors used Galois theory to develop a general verification lemma for $3\times 3$ matrices. We can expand on the method used there to help verify our case of higher dimensions.

\section{The \texorpdfstring{$\gds$}{GDS} Framework}\label{sec:GDS}
In this section, we formally propose the \textit{$\gds$ (Counting General Dominating Set) framework}, and then introduce a special type of signatures used in an important connection between $\gds$ and $\hol$.

\subsection{Definition}\label{sec:GDSdef}
We first introduce a variant of signature grids. A \textit{vertex-assignment signature grid} over $\mathcal{F}$, denoted by $\sg_V(G, \pi)$, is defined in the same way as a standard signature grid except for the following modifications: $\pi$ assigns to each vertex $v\in V$ a signature $f_v\in\mathcal{F}$ and an ordered list $l(v)=(u_1, \dots, u_{\degree{v}})$ of the vertices adjacent to $v$, where each adjacent vertex $u$ occurs exactly as many times as the edge multiplicity between $u$ and $v$ (except when $u=v$, in which case a self-loop contributes two entries of $v$ to $l(v)$). The arity of $f_v$ is $\degree{v}+1$, and its variables correspond to $v$ and the entries of $l(v)$. Let $\sigma_V:V\rightarrow\{0, 1\}$ be an assignment on the vertices, instead of edges. The evaluation of the vertex-assignment signature grid is $\prod_{v\in V}f_v(\sigma_V|_{v+l(v)})$, where $v+l(v)$ denotes the list obtained by prepending $v$ to $l(v)$, and $\sigma_V|_{v+l(v)}$ denotes the tuple obtained by applying $\sigma_V$ to the entries of $v+l(v)$, preserving order and repetitions.

\begin{definition}[$\gds$ problems]\label{def:GDS}
A $\gds$ problem, $\gds(\mathcal{F})$, parameterized by a set $\mathcal{F}$ of Boolean complex-valued signatures of arity at least one, is defined as follows: Given an instance that is a vertex-assignment signature grid $\sg_V(G, \pi)$ over $\mathcal{F}$, the output is the partition function of $\sg_V$,
\[
\gds_{\sg_V}=\sum_{\sigma_V:V\rightarrow\{0, 1\}}\prod_{v\in V}f_v(\sigma_V|_{v+l(v)}).
\]
The bipartite $\gds$ problem $\gds(\mathcal{F}\mid\mathcal{G})$ is defined similarly as in $\hol$.
\end{definition}

In this framework, we can also use the previously defined way of denoting symmetric signatures, with each variable having the same status as the others regardless of whether it corresponds to the vertex itself or those adjacent to it. For example, $\prob{\#3R-Dominating Set}$ is $\gds([0, 1, 1, 1, 1])$.

For the next two subsections, the rest of the paper does not depend on Propositions~\ref{prop:holtogds} to \ref{prop:vctogds2}, but the notation defined in Section~\ref{sec:apart} and Theorem~\ref{thm:times} are vital for Section~\ref{sec:uniform}.

\subsection{Equivalences with \texorpdfstring{$\hol$}{Holant}}
We observe that a $\hol$ problem on general graphs is equivalent to a bipartite $\gds$ problem, while a $\gds$ problem on general graphs is equivalent to a domain-4 bipartite $\hol$ problem.

\begin{proposition}[$\hol$ to $\gds$]\label{prop:holtogds}
Given a $\hol(\mathcal{F})$, where $\mathcal{F}$ is a set of Boolean signatures, there exists a Boolean signature set $\mathcal{G}$ such that
\[
\hol(\mathcal{F})\equiv_T\gds(\mathcal{G}\mid[1, 1, 1, 1]).
\]
$\mathcal{G}$ can be constructed in the following way. For each signature $f\in\mathcal{F}$, add a signature $g$ to $\mathcal{G}$, where $\ari{g}=\ari{f}+1$, and for each $\alpha$ such that $f(\alpha)=b$, let $g(0\alpha)=b$ and $g(1\alpha)=0$.
\end{proposition}
\begin{proof}
Given an instance of $\hol(\mathcal{F})$, we show that it can be converted to an instance of $\gds(\mathcal{G}\mid[1, 1, 1, 1])$ with the same output. The other direction is the reverse of this process.

Let $G=(V, E)$ be the graph of this instance's signature grid. We add a degree-2 vertex to every edge, and let the new vertices be $V'$. We use this new bipartite graph for the $\gds$ instance.

For every $v'\in V'$, we put $[1, 1, 1, 1]$. This means for any degree-2 vertex, no matter how it and its adjacent vertices are assigned, its own signature output would not affect the grid evaluation, so that it plays the role of a simple edge assignment in $\hol$.

For every $v\in V$ with signature $f$ in $\hol$, we put $g$ as described above. This way, we practically force the original vertices to be assigned 0, so that their signatures in $\gds$ behave the same as in $\hol$ without being affected by their own assignment.
\end{proof}

For the following equivalence, we define $\neq_{(12)}$ to be the domain-4 binary signature that has the following signature matrix:
\[
M_{\neq_{(12)}}=\begin{pmatrix}
    1 & 0 & 0 & 0\\
    0 & 0 & 1 & 0\\
    0 & 1 & 0 & 0\\
    0 & 0 & 0 & 1\\
\end{pmatrix}.
\]

\begin{proposition}[$\gds$ to domain-4 $\hol$]\label{prop:gdstohol4}
Given $\gds(\mathcal{F})$, where $\mathcal{F}$ is a set of Boolean signatures, there exists a domain-4 signature set $\mathcal{G}$ such that
\[
\gds(\mathcal{F})\equiv_T\hol_4(\mathcal{G}\mid\neq_{(12)}).
\]
$\mathcal{G}$ can be constructed in the following way. For each signature $f\in\mathcal{F}$, if $\ari{f}=1$, no action is taken; otherwise add a signature $g$ to $\mathcal{G}$, where $\ari{g}=\ari{f}-1$, and $g$ satisfies:
\[
g(\alpha)=\begin{cases}
f(0\alpha) & \alpha\in\{0, 1\}^n\\
f(1(\alpha-\mathbf{2})) & \alpha\in\{2, 3\}^n\\
0 & \text{otherwise}
\end{cases},
\]
where $n=\ari{f}-1$, and $\alpha-\mathbf{2}$ is the string over $\{0, 1\}$ obtained by subtracting $2$ from every digit of $\alpha$.
\end{proposition}
\begin{proof}
Basically, given an instance of the $\gds$ problem, we take its graph and assign each edge so that it encodes the assignment to both its incident vertices in $\gds$. From every vertex's perspective, if it is assigned 0 (1, resp.) in $\gds$, any of its incident edges should be assigned either 0 or 1 (2 or 3, resp.) in $\hol_4$, depending on whether the other incident vertex is assigned 0 or 1 in $\gds$.

The signature output at any non-isolated vertex is then determined first by whether the input string $\alpha$ is in $\{0, 1\}^n$ or $\{2, 3\}^n$, which corresponds to whether the vertex is assigned 0 or 1 in $\gds$, and then by the specific contents of $\alpha$. An isolated vertex does not admit this initial distinction, but it can be removed from the problem without affecting the complexity, as any unary $\gds$ signature $f$ on an isolated vertex contributes only a factor of $f(0)+f(1)$ to the $\gds$ value. 

Note, however, that this edge assignment introduces an asymmetry: if two adjacent vertices are assigned differently in $\gds$, the edge in $\hol_4$ should be assigned 1 from one vertex's perspective, but 2 from the other's. To account for this discrepancy, we put an additional $\neq_{(12)}$ on every edge.
\end{proof}

The signature $\neq_{(12)}$ can be converted to the domain-4 binary equality signature through holographic transformation (see \cite[Section~1.3.2]{cai2017complexity} for definition), making any $\gds$ problem also equivalent to a non-bipartite $\hol_4$ problem.

\subsection{Writing It Apart}\label{sec:apart}
We write an arity-$n$ signature $f$ as $\begin{pmatrix}
    f_0\\
    f_1
\end{pmatrix}$, where $f_0$ and $f_1$ are both $(n-1)$-arity signatures, if $f$ satisfies $f(0, x_1, \dots, x_{n-1})=f_0(x_1, \dots, x_{n-1})$ and $f(1, x_1, \dots, x_{n-1})=f_1(x_1, \dots, x_{n-1})$. We may call $f_0$ and $f_1$ \textit{sub-signatures} of $f$.

We derive some immediate facts from the \#P-hardness of \prob{\#3RPB-Vertex Cover} by Xia, Zhang and Zhao~\cite{xia2007vcmatching} using this notation.

\begin{proposition}\label{prop:vctogds1}
\prob{\#3RPB-Vertex Cover} can be reduced to $\plgds(\begin{pmatrix}
    [0, 0, 0, 1]\\
    [1, 1, 1, 1]
\end{pmatrix}\mid\begin{pmatrix}
    [0, 0, 0, 1]\\
    [1, 1, 1, 1]
\end{pmatrix})$, which is therefore \#P-hard.
\end{proposition}
\begin{proof}
In this $\gds$ problem, if a vertex is assigned 0, all of its adjacent vertices have to be assigned 1 for the vertex to have a non-zero contribution; if a vertex is assigned 1, its adjacent vertices can be freely assigned without affecting the evaluation of this vertex \textemdash\ a clear equivalence to vertex cover.
\end{proof}

\begin{proposition}\label{prop:vctogds2}
\prob{\#3RPB-Vertex Cover} can be reduced to $\plgds(\begin{pmatrix}
    [1, 0, 0, 0]\\
    [1, 0, 0, 0]
\end{pmatrix}\mid\begin{pmatrix}
    [1, 1, 1, 2]\\
    [1, 1, 1, 2]
\end{pmatrix})$, which is therefore \#P-hard.
\end{proposition}
\begin{proof}
In this $\gds$ problem, due to the signatures on the LHS, all vertices on the RHS have to be assigned 0 for a non-zero contribution, while the assignment to every vertex on the LHS still corresponds to whether or not it is in the vertex cover.

Fix an assignment to the LHS. In the original vertex cover instance, for each vertex on the RHS, if any of its adjacent vertices is assigned 0, the vertex must belong to the cover; if all of its adjacent vertices are assigned 1, the vertex may either belong or not belong to the cover. These one or two choices are reflected in the signature outputs on the RHS, rather than in the numerical assignment to these vertices, which remains 0.

All signature outputs on the LHS are then 1, while the product of the RHS signature outputs counts the vertex covers extending the fixed LHS assignment. Summing over all LHS assignments therefore gives exactly the number of vertex covers.
\end{proof}

Using this notation, with both problems restricted to simple graphs, we may express any \prob{\#$(\sigma, \rho)$-Set} problem as a $\gds$ problem where every signature can be written as $\begin{pmatrix}
    f_\rho\\
    f_\sigma
\end{pmatrix}$, in which $f_\sigma$ and $f_\rho$ are symmetric signatures, and for every $x\in\sigma$ ($x\notin\sigma$, resp.), the value of $f_\sigma$ on inputs with Hamming weight $x$ (if allowed) is 1 (0, resp.), similarly for $f_\rho$. For example, $\prob{\#3R-Total Dominating Set}$ is represented by $\gds(\begin{pmatrix}
    [0, 1, 1, 1]\\
    [0, 1, 1, 1]
\end{pmatrix})$.

We may refer to signatures of the form $\begin{pmatrix}
    f_0\\
    f_0
\end{pmatrix}$ as \textit{uniform} signatures. Given a signature set $\mathcal{F}_0$, we define $\begin{pmatrix}
    \mathcal{F}_0\\
    \mathcal{F}_0
\end{pmatrix}$ as $\{f\mid f=\begin{pmatrix}
    f_0\\
    f_0
\end{pmatrix}, f_0\in\mathcal{F}_0\}$. For a vertex $v$ with uniform signature $f_v$, we write $f_v=\begin{pmatrix}
    f_{v, 0}\\
    f_{v, 0}
\end{pmatrix}$. Then we have $f_v(\sigma_V|_{v+l(v)})=f_{v, 0}(\sigma_V|_{l(v)})$ regardless of how $v$ is assigned numerically.

We then state an important factorization that, under suitable conditions, expresses a $\gds$ value as a product of two $\hol$ values.

\begin{theorem}\label{thm:times}
Given a vertex-assignment signature grid $\sg=(G, \pi)$, where $G=(U, V, E)$ is a bipartite graph with no isolated vertices, and signature sets $\mathcal{F}=\begin{pmatrix}
    \mathcal{F}_0\\
    \mathcal{F}_0
\end{pmatrix}$, $\mathcal{G}=\begin{pmatrix}
    \mathcal{G}_0\\
    \mathcal{G}_0
\end{pmatrix}$, we have the following: \footnote{We denote the vertex-assignment signature grid as $\sg$ for simplicity. Although slightly repetitive, we use $\gds_\sg(\mathcal{F})$ to denote the value of $\gds(\mathcal{F})$ over the instance $\sg$ for clarity; similarly for $\hol$ or the bipartite $(\mathcal{F}\mid\mathcal{G})$.}
\[
\gds_{\sg}(\begin{pmatrix}
    \mathcal{F}_0\\
    \mathcal{F}_0
\end{pmatrix}\mid\begin{pmatrix}
    \mathcal{G}_0\\
    \mathcal{G}_0
\end{pmatrix})=\hol_{\sg_1}(\mathcal{F}_0\mid\mathcal{EQ})\times\hol_{\sg_2}(\mathcal{EQ}\mid\mathcal{G}_0)
\]
where $\sg_1$ and $\sg_2$ are $\hol$ signature grids derived by replacing every $\begin{pmatrix}
    f_0\\
    f_0
\end{pmatrix}$ with $f_0$ on one side, and substituting every signature with an equality signature on the other.
\end{theorem}
\begin{proof}
Write $\alpha:U\rightarrow\{0, 1\}$ and $\beta:V\rightarrow\{0, 1\}$. Based on Definition~\ref{def:GDS}, we have:
\begin{align}
\gds_{\sg}&=\sum_{\alpha:U\rightarrow\{0, 1\}}\sum_{\beta:V\rightarrow\{0, 1\}}\prod_{u\in U}f_{u, 0}(\beta|_{l(u)})\prod_{v\in V}f_{v, 0}(\alpha|_{l(v)})\\
&=\sum_{\alpha:U\rightarrow\{0, 1\}}\left(\left(\prod_{v\in V}f_{v, 0}(\alpha|_{l(v)})\right)\left(\sum_{\beta:V\rightarrow\{0, 1\}}\prod_{u\in U}f_{u, 0}(\beta|_{l(u)})\right)\right)\label{eq:step2}\\
&=\left(\sum_{\beta:V\rightarrow\{0, 1\}}\prod_{u\in U}f_{u, 0}(\beta|_{l(u)})\right)\left(\sum_{\alpha:U\rightarrow\{0, 1\}}\prod_{v\in V}f_{v, 0}(\alpha|_{l(v)})\right)\label{eq:step3}\\
&=\hol_{\sg_1}\times\hol_{\sg_2}.
\end{align}
Step \eqref{eq:step2} is because the entries of $l(v)$ do not overlap with $V$, and therefore $\prod_{v\in V}f_{v, 0}(\alpha|_{l(v)})$ is independent of $\beta$ and can be extracted as a common factor; similarly for step \eqref{eq:step3}.
\end{proof}

\begin{remark}
If $G$ has $r$ isolated vertices, the RHS of the equality requires an additional factor of $2^r$ (interpreting $=_0$ as the signature $[1]$), since the two choices for the Boolean assignment to each isolated vertex in the $\gds$ instance are not reflected in the edge assignments in the $\hol$ instances.
\end{remark}

\section{From Gadget Construction to \prob{\#Dominating Set}}\label{sec:gadget}
In this section, we adapt the technique of gadget construction to the $\gds$ framework. Unlike gadget construction in $\hol$ which cleanly separates external and internal edges, the formulation proves to be more complicated in $\gds$. Combining gadget construction with polynomial interpolation can impose specific restrictions on the assignment of vertices (sometimes similar to a simulation of logic gates), as is demonstrated by a polynomial reduction from \prob{\#3RP-Matching} to \prob{\#3RPBS-Dominating Set}.

\subsection{Gadget Construction in \texorpdfstring{$\gds$}{\#GDS}}
\begin{figure}[t!]
    \centering
    \includegraphics[width=4cm]{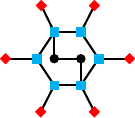}
    \caption{An example gadget. Diamonds are external vertices ($V_E$), squares are bridging vertices ($V_B$), and circles are internal vertices ($V_I$).}
    \label{pic:gadget}
\end{figure}

A $\gds$ gadget consists of three types of vertices. We denote the vertices that may connect to vertices outside the gadget as $V_E$, the \textit{external vertices}. Among the non-external vertices, those adjacent to at least one external vertex are $V_B$, the \textit{bridging vertices}; others not adjacent to any external vertex are $V_I$, the \textit{internal vertices} (e.g. Figure \ref{pic:gadget}).

All vertices are assigned Boolean values. All non-external vertices also have specified signatures, and their edges are explicitly defined within the gadget. External vertices $V_E$ are treated differently. Since $V_E$ fixes the gadget boundary, the gadget has no access to vertices beyond that boundary. In particular, if a vertex is adjacent to an external vertex but lies beyond the boundary, the gadget cannot take its Boolean assignment as input. Therefore, we may not yet be able to compute the signature outputs of external vertices. To keep the gadget reusable, we do not specify the dangling edges and the signatures of external vertices; these can be chosen when the gadget is connected to other parts.

For a gadget, we treat the Boolean assignment to $V_E$ and $V_B$ as input, and sum over all assignments to $V_I$. We call the resulting function the gadget's \textit{gadgeture} (gadget-signature) to distinguish it from an ordinary vertex signature, which is simpler in definition. (Note that $\hol$ does not need this distinction, since a $\hol$ gadget's signature also takes the assignment to external (dangling) edges as input and outputs a value, same as a vertex's signature.)

Let $\abs{V_E}=a$, $\abs{V_B}=b$, and let $\{x_1, \dots, x_a\}$ and $\{y_1, \dots, y_b\}$ be the assignment to $V_E$ and $V_B$, respectively. The gadgeture is
\[
g(x_1, \dots, x_a, y_1, \dots, y_b)=\sum_{\sigma_V:V_I\rightarrow\{0, 1\}}\prod_{v\in V_B\cup V_I}f_v\left(\hat{\sigma_V}|_{v+l(v)}\right)
\]
where $\hat{\sigma_V}:V_E\cup V_B\cup V_I\rightarrow\{0, 1\}$ extends $\sigma_V$ by incorporating $\{x_1, \dots, x_a, y_1, \dots, y_b\}$.

The signature outputs of external vertices are excluded because they may not yet be determined. In addition to the assignment to external vertices, which may be read by the bridging vertices to produce their outputs, the assignment to bridging vertices must also remain as input to the gadgeture, because when the gadget is connected to other parts, the signatures of external vertices will need to be specified and evaluated, and their outputs depend on the assignment of adjacent bridging vertices. This dependency is essential for later sum-of-product computations, and cannot be eliminated by summing over bridging vertices at this stage.

\begin{figure}[t!]
    \centering
    \includegraphics[width=8cm]{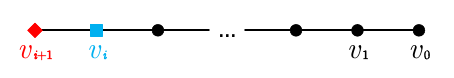}
    \caption{A simple gadget.}
    \label{pic:simple}
\end{figure}

For an example process of gadget construction and how a linear recurrence between gadgeture vectors can be derived, we look at the formulation of the gadget shown in Figure \ref{pic:simple}. At the beginning, we construct $H_1$: $V_E=\{v_1\}$, $V_B=\{v_0\}$, $V_I=\emptyset$, with an edge between $v_1$ and $v_0$. The signature put on $v_0$ is $[0, 1, 1]$ from \prob{\#Dominating Set}. Suppose $z_1$ and $z_0$ are the Boolean assignment to $v_1$ and $v_0$ respectively. The gadgeture $g_1(z_1, z_0)$ of $H_1$ is:
\[
M_{g_1(z_1, z_0)}=\begin{pmatrix}
    g_{1, 00}\\
    g_{1, 01}\\
    g_{1, 10}\\
    g_{1, 11}
\end{pmatrix}=\begin{pmatrix}
    0\\
    1\\
    1\\
    1
\end{pmatrix},
\]
where we write $g_i(b_1, b_2)$ as $g_{i, b_1b_2}$ for simplicity.

Suppose we now have $H_i$, with gadgeture $g_i$. We construct $H_{i+1}$ by adding a new external vertex $v_{i+1}$ with the Boolean assignment $z_{i+1}$, and connecting $v_{i+1}$ to $v_i$, making $v_i$ a bridging vertex, while the rest are internal vertices. We place $[0, 1, 1, 1]$ on $v_i$. Since it is now made a bridging vertex, its signature output needs to be properly included in $H_{i+1}$'s gadgeture $g_{i+1}$. We observe the following linear relation between $g_{i+1}$ and $g_i$:
\[
M_{g_{i+1}(z_{i+1}, z_i)}=\begin{pmatrix}
    g_{i+1, 00}\\
    g_{i+1, 01}\\
    g_{i+1, 10}\\
    g_{i+1, 11}
\end{pmatrix}=\begin{pmatrix}
    0 & 1 & 0 & 0\\
    0 & 0 & 1 & 1\\
    1 & 1 & 0 & 0\\
    0 & 0 & 1 & 1
\end{pmatrix}\begin{pmatrix}
    g_{i, 00}\\
    g_{i, 01}\\
    g_{i, 10}\\
    g_{i, 11}
\end{pmatrix}.
\]
To demonstrate, we look at how the topmost element, $g_{i+1, 00}$, is derived. This asks for both $v_{i+1}$ and $v_i$ to be assigned 0, while $v_i$'s signature output is multiplied into $g_i$ to form the value of $g_{i+1}$. If $v_{i-1}$ is also assigned 0, $v_i$'s signature output would be 0, i.e. there would be no vertex that dominates $v_i$. If $v_{i-1}$ is assigned 1, $v_i$'s signature output would be 1. As $v_i$ cannot be assigned 0 and 1 at the same time, the rest of the cases are not valid.

\subsection{The Ladder Gadget}\label{sec:equalor}
\begin{figure}[t!]
    \centering
    \includegraphics[width=10cm]{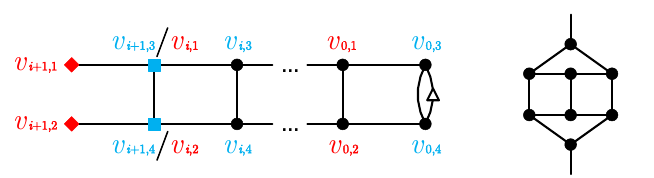}
    \caption{\textit{Left}: The ladder gadget. \textit{Right}: The structure represented by the triangle in the ladder gadget. This is our go-to method of construction for reducing to 3-regular bipartite simple graphs, taking inspiration from \protect\cite{xia2007vcmatching}.}
    \label{pic:ladder}
\end{figure}

We now present a gadget shown in Figure \ref{pic:ladder}, which we call the \textit{ladder gadget}.

All vertices have (or will have) the signature of $[0, 1, 1, 1, 1]$. $z_{i, j}$ is the Boolean assignment to $v_{i, j}$ for every $i, j$.

We construct a series of ladder gadgets, $H_0, H_1, \dots$, with $H_i$ having the gadgeture $g_i$. We write $g_i(b_1, b_2, b_3, b_4)$ as $g_{i, b_1b_2b_3b_4}$ for simplicity.

The initial gadget $H_0$ has the basic structure of $V_E=\{v_{0, 1}, v_{0, 2}\}$, $V_B=\{v_{0, 3}, v_{0, 4}\}$ and $V_I$ being the vertices in the structure represented by the triangle. Then,
\[
\begin{split}
    &M_{g_0(z_{0, 1}, z_{0, 2}, z_{0, 3}, z_{0, 4})}\\
    ={}&\begin{pmatrix}
        g_{0, 0000} & g_{0, 0001} & g_{0, 0010} & \dots & g_{0, 1111}
    \end{pmatrix}^\T\\
    ={}&\begin{pmatrix}
        57 & 179 & 179 & 194 & 96 & 179 & 179 & 194 & 96 & 179 & 179 & 194 & 165 & 179 & 179 & 194
    \end{pmatrix}^\T,
\end{split}
\]
which can be calculated with Mathematica using the code in Appendix \ref{app:code}.

Then, when we derive $H_{i+1}$ from $H_i$, we add $v_{i+1, 1}$ and $v_{i+1, 2}$, while letting $v_{i, 1}$ and $v_{i, 2}$ take the place of $v_{i+1, 3}$ and $v_{i+1, 4}$. The full form of the linear recurrence of gadgetures, $g_{i+1}=Ag_i$, is as follows:
\[
\begin{pmatrix}
    g_{i+1, 0000}\\
    g_{i+1, 0001}\\
    g_{i+1, 0010}\\
    g_{i+1, 0011}\\
    g_{i+1, 0100}\\
    g_{i+1, 0101}\\
    g_{i+1, 0110}\\
    g_{i+1, 0111}\\
    g_{i+1, 1000}\\
    g_{i+1, 1001}\\
    g_{i+1, 1010}\\
    g_{i+1, 1011}\\
    g_{i+1, 1100}\\
    g_{i+1, 1101}\\
    g_{i+1, 1110}\\
    g_{i+1, 1111}
\end{pmatrix}
=\begin{pmatrix}
    0 & 0 & 0 & 1 & 0 & 0 & 0 & 0 & 0 & 0 & 0 & 0 & 0 & 0 & 0 & 0\\
    0 & 0 & 0 & 0 & 1 & 1 & 1 & 1 & 0 & 0 & 0 & 0 & 0 & 0 & 0 & 0\\
    0 & 0 & 0 & 0 & 0 & 0 & 0 & 0 & 1 & 1 & 1 & 1 & 0 & 0 & 0 & 0\\
    0 & 0 & 0 & 0 & 0 & 0 & 0 & 0 & 0 & 0 & 0 & 0 & 1 & 1 & 1 & 1\\
    0 & 0 & 1 & 1 & 0 & 0 & 0 & 0 & 0 & 0 & 0 & 0 & 0 & 0 & 0 & 0\\
    0 & 0 & 0 & 0 & 1 & 1 & 1 & 1 & 0 & 0 & 0 & 0 & 0 & 0 & 0 & 0\\
    0 & 0 & 0 & 0 & 0 & 0 & 0 & 0 & 1 & 1 & 1 & 1 & 0 & 0 & 0 & 0\\
    0 & 0 & 0 & 0 & 0 & 0 & 0 & 0 & 0 & 0 & 0 & 0 & 1 & 1 & 1 & 1\\
    0 & 1 & 0 & 1 & 0 & 0 & 0 & 0 & 0 & 0 & 0 & 0 & 0 & 0 & 0 & 0\\
    0 & 0 & 0 & 0 & 1 & 1 & 1 & 1 & 0 & 0 & 0 & 0 & 0 & 0 & 0 & 0\\
    0 & 0 & 0 & 0 & 0 & 0 & 0 & 0 & 1 & 1 & 1 & 1 & 0 & 0 & 0 & 0\\
    0 & 0 & 0 & 0 & 0 & 0 & 0 & 0 & 0 & 0 & 0 & 0 & 1 & 1 & 1 & 1\\
    1 & 1 & 1 & 1 & 0 & 0 & 0 & 0 & 0 & 0 & 0 & 0 & 0 & 0 & 0 & 0\\
    0 & 0 & 0 & 0 & 1 & 1 & 1 & 1 & 0 & 0 & 0 & 0 & 0 & 0 & 0 & 0\\
    0 & 0 & 0 & 0 & 0 & 0 & 0 & 0 & 1 & 1 & 1 & 1 & 0 & 0 & 0 & 0\\
    0 & 0 & 0 & 0 & 0 & 0 & 0 & 0 & 0 & 0 & 0 & 0 & 1 & 1 & 1 & 1\\
\end{pmatrix}\begin{pmatrix}
    g_{i, 0000}\\
    g_{i, 0001}\\
    g_{i, 0010}\\
    g_{i, 0011}\\
    g_{i, 0100}\\
    g_{i, 0101}\\
    g_{i, 0110}\\
    g_{i, 0111}\\
    g_{i, 1000}\\
    g_{i, 1001}\\
    g_{i, 1010}\\
    g_{i, 1011}\\
    g_{i, 1100}\\
    g_{i, 1101}\\
    g_{i, 1110}\\
    g_{i, 1111}
\end{pmatrix}
\]

Similar to a technique used in \cite{xia2007vcmatching}, we observe that there exist some groups of elements in the gadgeture vector, where the elements within each group always equal one another for every $i$. Each of $g_{i, 0000}$ and $g_{i, 1100}$ forms its own group, while the other groups are:
\begin{align*}
\begin{split}
    &g_{i, 0001}=g_{i, 0010}=g_{i, 0101}=g_{i, 0110}\\
    &=g_{i, 1001}=g_{i, 1010}=g_{i, 1101}=g_{i, 1110};
\end{split}\\
    &g_{i, 0011}=g_{i, 0111}=g_{i, 1011}=g_{i, 1111};\\
    &g_{i, 0100}=g_{i, 1000}.
\end{align*}
This means there are only five distinct values within each gadgeture vector. We can rewrite the initial vector and the recurrence as follows:
\begin{gather}
    \begin{pmatrix}
        g_{0, 0000}\\
        g_{0, 0001}\\
        g_{0, 0011}\\
        g_{0, 0100}\\
        g_{0, 1100}
    \end{pmatrix}=\begin{pmatrix}
        57\\
        179\\
        194\\
        96\\
        165
    \end{pmatrix}; \begin{pmatrix}
        g_{i+1, 0000}\\
        g_{i+1, 0001}\\
        g_{i+1, 0011}\\
        g_{i+1, 0100}\\
        g_{i+1, 1100}
    \end{pmatrix}=\begin{pmatrix}
        0 & 0 & 1 & 0 & 0\\
        0 & 2 & 1 & 1 & 0\\
        0 & 2 & 1 & 0 & 1\\
        0 & 1 & 1 & 0 & 0\\
        1 & 2 & 1 & 0 & 0
    \end{pmatrix}\begin{pmatrix}
        g_{i, 0000}\\
        g_{i, 0001}\\
        g_{i, 0011}\\
        g_{i, 0100}\\
        g_{i, 1100}
    \end{pmatrix}.\label{eq:dsrecur}
\end{gather}
We denote the $5\times 5$ recurrence matrix above as $A_5$.

Suppose we have a graph which includes $n$ copies of $H_s$, with no non-external vertex shared between copies. Then the number of dominating sets can be expressed as the polynomial
\begin{equation}\label{eq:dspoly}
\sum_{i_1+i_2+i_3+i_4+i_5=n}c_{i_1, i_2, i_3, i_4, i_5}g_{s, 0000}^{i_1}g_{s, 0001}^{i_2}g_{s, 0011}^{i_3}g_{s, 0100}^{i_4}g_{s, 1100}^{i_5}.
\end{equation}

In order to retrieve the coefficients of \eqref{eq:dspoly}, the recurrence \eqref{eq:dsrecur} needs to satisfy the requirements in Lemma~\ref{lem:interpolation}. The first two requirements can be easily verified. We then prove that the third requirement is also fulfilled. The proof uses some basic Galois theory.

\begin{lemma}\label{lem:5times5}
For all $(d_1, d_2, d_3, d_4, d_5)\in\mathbb{Z}^5-\{(0, 0, 0, 0, 0)\}$ with $d_1+d_2+d_3+d_4+d_5=0$, $\lambda_1^{d_1}\lambda_2^{d_2}\lambda_3^{d_3}\lambda_4^{d_4}\lambda_5^{d_5}\neq 1$, where $\lambda_1, \dots, \lambda_5$ are distinct eigenvalues of $A_5$.
\end{lemma}
\begin{proof}
The characteristic polynomial of $A_5$ is $f(x)=-x^5+3x^4+2x^3+2x^2-x-1$, which can be verified to be irreducible over $\mathbb{Q}$, and have exactly two non-real complex roots. It is known that the Galois group $G$ of a degree-$p$ irreducible polynomial (where $p$ is a prime number) which has exactly two non-real complex roots is isomorphic to the symmetric group $S_p$ (see \cite[Chapter~II, Theorem~8]{jacobson1964galois}).
    
If there did exist a non-trivial assignment to $(d_1, d_2, d_3, d_4, d_5)$ which lets $\lambda_1^{d_1}\lambda_2^{d_2}\lambda_3^{d_3}\lambda_4^{d_4}\lambda_5^{d_5}=1$, there must be at least two among the five exponents that are different. Due to the Galois group being $S_p$, we could assume without loss of generality that $d_1\neq d_2$, for otherwise if it is another pair that is different, say $d_3\neq d_4$, we can use the $\mathbb{Q}$-automorphism (automorphism of the splitting field of $f$ over $\mathbb{Q}$ which fixes $\mathbb{Q}$ pointwise) $g$ out of $G$ which has $g(\lambda_3)=\lambda_1$ and $g(\lambda_4)=\lambda_2$ to permute the roots to achieve the desired form. It is also due to the ability to freely permute the roots that we would have $\lambda_2^{d_1}\lambda_1^{d_2}\lambda_3^{d_3}\lambda_4^{d_4}\lambda_5^{d_5}=1$, which means $\frac{\lambda_1}{\lambda_2}$ would be a root of unity. Similarly, we could get that for every $s, t=1, \dots, 5$ where $s\neq t$, $\frac{\lambda_s}{\lambda_t}$ must be a root of unity, which is false, since there exist $s\neq t$ such that $\abs{\frac{\lambda_s}{\lambda_t}}\neq 1$.
\end{proof}

\subsection{Reducing to \prob{\#Dominating Set}}\label{app:mainds}
\begin{figure}[t!]
    \centering
    \includegraphics[height=4cm]{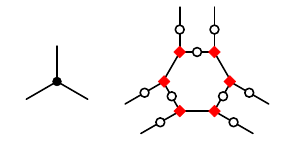}
    \caption{The reduction from counting matchings to counting dominating sets. Hollow circles represent ladder gadgets, with each diamond being an external vertex of two ladder gadgets.}
    \label{pic:reduction}
\end{figure}

Using the ladder gadget, we can prove Theorem~\ref{thm:mainds}.

\begin{proof}[Proof of Theorem~\ref{thm:mainds}]
Let $G=(V, E)$ be a 3-regular planar graph with $n$ vertices. As shown in Figure \ref{pic:reduction}, we convert every vertex of $G$ into a hexagon, and every edge of $G$ into a pair of edges connecting pairs of adjacent vertices; then, we substitute certain edges with copies of the ladder gadget $H_s$ for the same $s$, using the vertices of the hexagons as external vertices. The resulting graph $G'_s$ is also 3-regular and planar, and does not contain any cycle of odd length or multiple edges. Therefore, it is a 3-regular planar bipartite simple graph.

$G'_s$ now contains $6n$ copies of $H_s$. As expected of the ladder gadgets, the vertex-assignment signature grid $\sg_s$ puts $[0, 1, 1, 1, 1]$ on all vertices (including the external vertices of $H_s$), so that the $\gds$ value is exactly the number of dominating sets of $G'_s$. The $\gds$ value of $\sg_s$ can then be written as 
\[
\sum_{i_1+i_2+i_3+i_4+i_5=6n}c_{i_1, i_2, i_3, i_4, i_5}g_{s, 0000}^{i_1}g_{s, 0001}^{i_2}g_{s, 0011}^{i_3}g_{s, 0100}^{i_4}g_{s, 1100}^{i_5}.
\]

By computing $\gds$ values for a sufficiently large, polynomially bounded number of instances $\sg_s$ with distinct $s$, and doing polynomial interpolation (which is feasible due to Lemma~\ref{lem:interpolation}), we are able to retrieve the coefficients. Specifically, we sum all coefficients $c_{i_1, 0, 0, 0, i_5}$ (where $i_1+i_5=6n$), which correspond precisely to the cases where, within each ladder gadget, both external vertices are assigned the same Boolean value, and both bridging vertices are assigned 0. In these cases, the six external vertices in every hexagon consist of three pairs, each equally assigned, and the four external vertices of every pair of ladder gadgets derived from the same edge in $G$ are equally assigned, representing an edge assignment in $G$.

Recall that the output of a gadgeture includes the signature outputs of bridging and internal vertices, but excludes that of external vertices. When evaluating the whole signature grid $\Omega_s$, we must therefore also include the signature output of each external vertex as a factor. Thus, an assignment can contribute a non-zero term only if all external vertices are dominated. We observe that the six external vertices in every hexagon are all dominated if and only if at least two of three pairs are assigned 1. This corresponds to selecting edges in the original $G$ such that at least two of three edges incident to every vertex are chosen. Formally, $\sum_{i_1+i_5=6n}c_{i_1, 0, 0, 0, i_5}$ equals the output of $\plhol([0, 0, 1, 1])$ (which is equivalent to $\plhol([1, 1, 0, 0])$ by flipping 0's and 1's in the edge assignment, proved \#P-hard as \prob{\#3RP-Matching} in \cite[Section~5]{xia2007vcmatching}) on $G$. Therefore, a polynomial reduction has been established from \prob{\#3RP-Matching} to \prob{\#3RPBS-Dominating Set}.
\end{proof}

\begin{remark}
If we lift the restriction to simple graphs, we can trivially reduce \prob{\#3RPBS-Dominating Set} to \prob{\#$d$RPB-Dominating Set} for every $d>3$ by adding $d-3$ parallel copies of each edge of a perfect matching.
\end{remark}

\begin{remark}
The coefficient sum used in the proof of Theorem~\ref{thm:mainds} imposes the XNOR condition on the assignment of the two external vertices of each ladder gadget, since it allows precisely the pairs 00 and 11. Other Boolean conditions can be simulated by selecting different coefficients. For example, for the polynomial \eqref{eq:dspoly}, where $n$ is the number of ladder gadget copies, summing all $c_{0, 0, 0, i_4, i_5}$ with $i_4+i_5=n$ allows precisely the external assignment pairs 01, 10 and 11, and therefore simulates OR. Similar approaches also apply to AND, NAND, NOR and XOR. When applying these simulations, we must also account for the signatures of the external vertices, whose outputs must be non-zero for an assignment to contribute.
\end{remark}

\section{From Powered \texorpdfstring{$\hol$}{Holant} to \prob{\#Total Dominating Set}}\label{sec:uniform}
In this section, we explore a specific case within uniform signatures that would allow the two $\hol$ values in Theorem~\ref{thm:times} to be equal, so that the $\gds$ value is the square of a $\hol$ value. This motivates powered $\hol$ problems and an applicable interpolation lemma. The study of this case contributes to the hardness proof of \prob{\#Total Dominating Set}, but the problem itself does not fall within the case; we instead obtain the required equality by restricting to symmetric bipartite graphs, and use it to prove Theorem~\ref{thm:maintds}.

\subsection{The Power of \texorpdfstring{$\hol$}{Holant}}\label{sec:holpow}
We introduce a few notations first.

\begin{definition}\label{def:holtri}
$\hol(=_d\mid M\mid =_d)$, where $d\in\mathbb{Z}^+$ and $M_{x_1, x_2}$ is a $2\times 2$ matrix, is a form of $\hol$ problem where $G=(U, V, W, E)$ is a tripartite graph, the vertices in $U\cup W$ are all assigned $=_d$ and no two of them are adjacent, and the vertices in $V$ are all of degree 2 and assigned $M$, with the edges corresponding to $x_1$ and $x_2$ always connecting to $U$ and $W$ respectively.
\end{definition}

Here $M$ directly refers to a signature. This class of problems is related to the edge-weighted version of counting graph homomorphisms~\cite{martin2000homomorphism}, but here $M$ is allowed to be asymmetric, as long as all copies of it are placed along the same direction in the graph, as specified in the final part of the definition.

Every $\hol(=_d\mid M\mid =_d)$ is equivalent to a $\hol(f\mid =_d)$, where $f=(=_d)M^{\otimes d}$, but not every symmetric signature $f$ admits such a representation; relevant details can be found in \cite[Lemma~2.13]{cai2017complexity}.

We use the following notation for fixed positive integer powers of $\hol$ evaluations.

\begin{definition}
$\hol^k(\mathcal{F})$, where $k$ is a fixed positive integer, is a problem that calculates for every instance $\sg$:
\[
(\hol^k)_\sg=(\hol_\sg)^k.
\]
\end{definition}

From Theorem~\ref{thm:times}, we then immediately get the following corollary:

\begin{corollary}\label{cor:gdstohol2}
Given a $d$-regular bipartite graph $G=(U, V, E)$ where $d\in\mathbb{Z}^+$, and a signature $f_0=(=_d)M^{\otimes d}$, where $M$ is symmetric, i.e. $M=\begin{pmatrix}
    x & y\\
    y & z
\end{pmatrix}$, or $M$ is asymmetric but $M(0, 0)=M(1, 1)$, i.e. $M=\begin{pmatrix}
    x & y\\
    z & x
\end{pmatrix}$, we have the following: \footnote{With a slight abuse of notation, if there is only one signature on each side, and the edge ordering does not matter or is already specified within the definition of the problem (e.g. in Definition~\ref{def:holtri}), we write $\gds_G$ and $\hol_G$ for the output of a given signature grid $\sg=(G, \pi)$.}
\[
\gds_G(\begin{pmatrix}
    f_0\\
    f_0
\end{pmatrix}\mid\begin{pmatrix}
    f_0\\
    f_0
\end{pmatrix})=(\hol^2)_G(f_0\mid =_d).
\]
\end{corollary}
\begin{proof}
If $M=\begin{pmatrix}
    x & y\\
    y & z
\end{pmatrix}$, both $\hol_G(f_0\mid =_d)$ and $\hol_G(=_d\mid f_0)$ are precisely $\hol_{G'}(=_d\mid M\mid =_d)$, where $G'$ is obtained from $G$ by inserting a new vertex on every edge. Attach the signature $M$ to every new vertex and replace every $f_0$ with $=_d$, and both instances are converted to the exact same signature grid, hence the same $\hol$ value.

If $M=\begin{pmatrix}
    x & y\\
    z & x
\end{pmatrix}$, $\hol_G(f_0\mid =_d)$ still equals $\hol_{G'}(=_d\mid M\mid =_d)$, and $\hol_G(=_d\mid f_0)=\hol_{G'}(=_d\mid M^\T\mid =_d)=\hol_{G'}((=_d)X^{\otimes d}\mid M^\T\mid X^{\otimes d}(=_d))=\hol_{G'}(=_d\mid XM^\T X\mid =_d)=\hol_{G'}(=_d\mid M\mid =_d)$, where $X$ stands for $\begin{pmatrix}
    0 & 1\\
    1 & 0
\end{pmatrix}$.
\end{proof}

This motivates us to study the complexity of $\hol^k$. If each signature of $\mathcal{F}$ is real-valued and of constant sign, it is clear that $\hol^k(\mathcal{F})\equiv_T\hol(\mathcal{F})$, but otherwise, it is not obvious that $\hol(\mathcal{F})\leq_T\hol^k(\mathcal{F})$. For such problems, a sufficient route is to reduce from a known \#P-hard $\hol^k$ problem to the target \textemdash\ i.e., to use a reduction between $\hol^k$ problems. A related strategy of adapting reduction chains to disjoint copies of instances was used by Huang and Lu~\cite[Lemma~IV.4]{huang2012weighted}. There are also other examples of interpolating powers or squared norms of polynomials using the corresponding degree bounds, such as in \cite[Section~9.1]{aaronson2013linearoptics}, \cite[Theorem~3.1]{przybylek2020phaseless} and \cite[Proof of Theorem~2.13]{cai2026plghnonneg}. The following lemma formalizes this elementary interpolation principle for powers of polynomials used in our reductions.

\begin{lemma}\label{lem:powerinterpol}
Fix an integer $k\geq 1$ and a number field $K$. Suppose that, for an input instance $I$, there is a polynomial $P_I\in K[x]$ of degree at most $d=\operatorname{poly}(\abs{I})$, and we can compute pairwise distinct points $x_0, \dots, x_{kd}\in K$ and construct corresponding instances $I_0, \dots, I_{kd}$ in polynomial time, such that the evaluation $Z(I_s)$ of each instance $I_s$ satisfies
\[
Z(I_s)=P_I(x_s), 
\]
where $s=0, \dots, kd$. Then an oracle for $(Z(I_s))^k$, with answers of length polynomial in $\abs{I}$, suffices to recover the polynomial $Q_I=(P_I)^k$ in polynomial time and to evaluate it at any specified $x\in K$ in time polynomial in $\abs{I}$ and the encoding length of $x$.
\end{lemma}
\begin{proof}
The polynomial $Q_I$ has degree at most $kd$. The oracle supplies
\[
(Z(I_s))^k=(P_I(x_s))^k=Q_I(x_s)
\]
at $kd+1$ pairwise distinct points. Polynomial interpolation therefore recovers $Q_I$, which we evaluate at $x$.
\end{proof}

Using this lemma, we can prove the following hardness lemma for Theorem~\ref{thm:maintds}.

\begin{lemma}\label{lem:plhol4}
Let $f=[603, -340, 115, -76]$. For $k\in\mathbb{Z}^+$, $\plhol^k(f\mid =_3)$ is \#P-hard.
\end{lemma}
\begin{proof}
We first normalize $f$ to $[1, a, b, c]$, where $a=-\frac{340}{603}, b=\frac{115}{603}, c=-\frac{76}{603}$, and connect it with $=_3$ by two parallel edges to obtain the signature $G_1=\begin{pmatrix}
    1 & b\\
    a & c
\end{pmatrix}$ straddled across sides. We then essentially follow the notation and construction in the proof of \cite[Lemma~2.2]{cai2023hol33} to interpolate the degenerate binary straddled signature $D$, with the distinction being that
\[
D=\begin{pmatrix}
    -x\\
    1
\end{pmatrix}\begin{pmatrix}
    -1 & y
\end{pmatrix}=\begin{pmatrix}
    x & -xy\\
    -1 & y
\end{pmatrix}=\begin{pmatrix}
    -x & y\\
    1 & 1
\end{pmatrix}\begin{pmatrix}
    1 & 0\\
    0 & 0
\end{pmatrix}\begin{pmatrix}
    -x & y\\
    1 & 1
\end{pmatrix}^{-1}.
\]
The symbols all represent fixed numbers in this case, and it can be checked that the conditions in \cite[Lemma~2.2]{cai2023hol33} are met. The interpolation works due to Lemma~\ref{lem:powerinterpol}; specifically, the powered evaluations satisfy
\[
(\hol^k)_{\Omega_s}=\sum_{i+j=kn}c_{i, j}'(\lambda^i\mu^j)^s
\]
for $s=0, \dots, kn$, and we need to obtain $c_{kn, 0}'$.

Let $g=[ay-1, by-a, cy-b]$, which can be obtained by connecting $f$ with $[-1, y]$. It can be verified that $\plhol(g\mid =_3)$ is not in the tractable cases of \cite[Theorem~1]{kowalczyk2016hol23} and is therefore \#P-hard. Since $g$ is positive, $\plhol^k(g\mid =_3)$ is also \#P-hard.

Consider an instance of $\plhol^k(g\mid =_3)$. Suppressing its degree-2 vertices gives a 3-regular planar multigraph. Through \cite[Theorem~3.2]{cai2023p3em}, which is known as the P3EM theorem, components other than $K_4$ and $M_{2, 3}$ (the multigraph on 2 vertices with 3 parallel edges) admit a planar 3-way edge matching \textemdash\ that is, the edges can be partitioned into triples such that, after simultaneously adding a new vertex for every triple and connecting it to the triple's three edge interiors, the resulting graph is still planar. Such a matching can be found in polynomial time. $K_4$ and $M_{2, 3}$ components have constant size, so we may compute their evaluations separately and multiply them into the final result.

For other components, we use a process similar to \cite[Lemma~4.7]{cai2023p3em}. Replace every degree-2 vertex by a degree-3 vertex assigned $f$, which retains its two original incident edges and connects to a copy of $[-1, y]$ with the third edge. This does not change the evaluation of the instance.

For each P3EM triple, add a component consisting of a new vertex assigned $=_3$ connecting to three vertices assigned $[-x, 1]$. This multiplies the evaluation by a non-zero computable factor.

Finally, substitute every pair of $[-x, 1]$ and $[-1, y]$ with the degenerate binary straddled signature $D$, and the P3EM theorem guarantees a way such that the resulting graph is still planar. This aligns with the orientation of $D$: the row and column inputs, corresponding to $[-x, 1]$ and $[-1, y]$ respectively, connect to $=_3$ on the RHS and $f$ on the LHS.

Since $D$ can be interpolated using only $f$ and $=_3$, we have established a polynomial reduction from $\plhol^k(g\mid =_3)$ to $\plhol^k(f\mid =_3)$.
\end{proof}

\subsection{\texorpdfstring{$\hol$}{Holant} for Symmetric Bipartite Graphs}
For \prob{\#3RB-Total Dominating Set}, $f_0=[0, 1, 1, 1]=(=_3)M^{\otimes 3}$, where $M=\begin{pmatrix}
    -1 & 0\\
    1 & 1
\end{pmatrix}$, which is not of either form allowed by Corollary~\ref{cor:gdstohol2}. If the input graph is unrestricted, $\hol_{\sg_1}$ and $\hol_{\sg_2}$ in Theorem~\ref{thm:times} are not guaranteed to be equal, so this factorization does not directly express the $\gds$ value as the square of a $\hol$ value. This gives us the idea of restricting the input graph to symmetric bipartite graphs, which have been defined in some research such as \cite{kakimura2010sym,cairns2015sym}.

\begin{definition}
Extending the definition in \cite{kakimura2010sym} to multigraphs, a bipartite graph $G=(U, V, E)$, where $U=\{u_1, \dots, u_n\}$ and $V=\{v_1, \dots, v_n\}$, is \textit{symmetric} if there exists a labeling of the vertices such that $\forall i, j\in\{1, \dots, n\}$, $m_G(u_i, v_j)=m_G(u_j, v_i)$, where $m_G(u, v)$ denotes the edge multiplicity between $u$ and $v$.

Fix such a vertex labeling, and number the edges between $u_i$ and $v_j$ as $\{e_{ij, 1}, \dots, e_{ij, m_G(u_i, v_j)}\}$. For every $t\in\{1, \dots, m_G(u_i, v_j)\}$, we call $e_{ji, t}$ the \textit{mirror edge} of $e_{ij, t}$, denoted by $e_{ij, t}^*$.

We write $\symgds(\mathcal{F}\mid\mathcal{G})$ and $\symhol(\mathcal{F}\mid\mathcal{G})$ for the problems of $\gds(\mathcal{F}\mid\mathcal{G})$ and $\hol(\mathcal{F}\mid\mathcal{G})$ on symmetric bipartite graphs.
\end{definition}

We then have the following lemma.

\begin{lemma}\label{lem:symswap}
For any signature sets $\mathcal{F}, \mathcal{G}$, and signature grids $\sg_1=(G, \pi_1)$ and $\sg_2=(G, \pi_2)$ where $G=(U=\{u_1, \dots, u_n\}, V=\{v_1, \dots, v_n\}, E)$ is a symmetric bipartite graph, and $\forall i\in\{1, \dots, n\}$, $\pi_1$ assigns to $u_i$ ($v_i$, resp.) the same signature as $\pi_2$ assigns to $v_i$ ($u_i$, resp.), with mirror edges occupying the same input positions, we have
\[
\symhol_{\sg_1}(\mathcal{F}\mid\mathcal{G})=\symhol_{\sg_2}(\mathcal{G}\mid\mathcal{F}).
\]
\end{lemma}
\begin{proof}
For any edge assignment $\sigma_1$ to $\sg_1$, there is a one-to-one correspondence to an edge assignment $\sigma_2$ to $\sg_2$ by letting $\sigma_2(e)=\sigma_1(e^*)$ for every $e\in E$. Then, denoting $f_i$ ($g_i$, resp.) as the signature from $\mathcal{F}$ ($\mathcal{G}$, resp.) assigned to $u_i$ in $\pi_1$ and $v_i$ in $\pi_2$ ($v_i$ in $\pi_1$ and $u_i$ in $\pi_2$, resp.), we have:
\[
\prod_{i=1}^nf_i(\sigma_1|_{E(u_i)})\prod_{j=1}^ng_j(\sigma_1|_{E(v_j)})=\prod_{i=1}^nf_i(\sigma_2|_{E(v_i)})\prod_{j=1}^ng_j(\sigma_2|_{E(u_j)}).
\]
By adding $\sum_{\sigma_1:E\rightarrow\{0, 1\}}$ to the front of the left-hand side of this equation, and $\sum_{\sigma_2:E\rightarrow\{0, 1\}}$ to the front of the RHS, we get the equation in the lemma.
\end{proof}

Then, combining Theorem~\ref{thm:times} and Lemma~\ref{lem:symswap}, we have the following corollary:

\begin{corollary}\label{cor:symgds}
Given a $d$-regular symmetric bipartite graph $G=(U, V, E)$ where $d\in\mathbb{Z}^+$, and an arity-$d$ symmetric signature $f_0$, we have the following:
\[
\symgds_G(\begin{pmatrix}
    f_0\\
    f_0
\end{pmatrix}\mid\begin{pmatrix}
    f_0\\
    f_0
\end{pmatrix})=(\symhol^2)_G(f_0\mid =_d).
\]
\end{corollary}

From this equality, we know that:
\begin{equation}
\symplsgds(\begin{pmatrix}
    [0, 1, 1, 1]\\
    [0, 1, 1, 1]
\end{pmatrix}\mid\begin{pmatrix}
    [0, 1, 1, 1]\\
    [0, 1, 1, 1]
\end{pmatrix})\equiv_T\symplshol^2([0, 1, 1, 1]\mid =_3).\label{eq:symequiv}
\end{equation}

Therefore, by proving the \#P-hardness of the latter problem, we can prove the \#P-hardness of \prob{\#3RPBS-Total Dominating Set}, in fact with another restriction of being on symmetric bipartite graphs. But it is also because of this restriction that we cannot directly invoke any hardness result obtained for $\hol^k$ on general graphs; we can however establish a reduction from it. An important part is taking any bipartite graph and constructing a symmetric bipartite graph. This is done in the following lemma:

\begin{lemma}\label{lem:holtosymhol}
For any $d\in\mathbb{Z}^+$, $f=(=_d)M^{\otimes d}$, where $M=\begin{pmatrix}
    x & y\\
    y & z
\end{pmatrix}$ or $\begin{pmatrix}
    x & y\\
    z & x
\end{pmatrix}$, and $k\in\mathbb{Z}^+$, the following polynomial reduction holds:
\[
\hol^{2k}(f\mid =_d)\leq_T\symhol^k(f\mid =_d),
\]
which also preserves graph planarity.
\end{lemma}
\begin{proof}
To construct a symmetric bipartite graph, we use the simple method of making another copy of the original graph, but exchanging the two sides before adding it to the original copy as a separate component. Formally, for $G=(U, V, E)$, we make a copy $G'=(V', U', E')$ where we put $V'$ and $U'$ on the LHS and RHS respectively without changing the edges, and denote $G^2$ as the graph $(U\cup V', V\cup U', E\cup E')$.

To see that $G^2$ is symmetric, we label $U=\{u_1, \dots, u_n\}$, $V=\{v_{n+1}, \dots, v_{n+m}\}$, $U'=\{v_1, \dots, v_n\}$ and $V'=\{u_{n+1}, \dots, u_{n+m}\}$ (notice the unusual labeling of $U'$ and $V'$), where $v_i$ is copied from $u_i$ for $i\in\{1, \dots, n\}$ and $u_j$ is copied from $v_j$ for $j\in\{n+1, \dots, n+m\}$. Then it can be easily seen that $m_{G^2}(u_i, v_j)=m_{G^2}(u_j, v_i)$ for $i, j\in\{1, \dots, n+m\}$.

The symmetry of $G^2$ is independent of its embedding. If $G$ is planar, drawing $G$ and $G'$ in disjoint regions of the plane gives a planar embedding of $G^2$.

We observe that $\hol_{G^2}(f\mid =_d)=\hol_G(f\mid =_d)\times\hol_{G'}(f\mid =_d)$ for any symmetric signature $f$. As in the proof of Corollary~\ref{cor:gdstohol2}, we know that $\hol_G(f\mid =_d)=\hol_{G'}(f\mid =_d)$ is guaranteed if $f=(=_d)M^{\otimes d}$, where $M=\begin{pmatrix}
    x & y\\
    y & z
\end{pmatrix}$ or $\begin{pmatrix}
    x & y\\
    z & x
\end{pmatrix}$. Therefore the reduction holds by considering $(\hol^{2k})_G(f\mid =_d)$ as $(\symhol^k)_{G^2}(f\mid =_d)$.
\end{proof}

\subsection{Reducing to \prob{\#Total Dominating Set}}
We now complete the reduction that proves Theorem~\ref{thm:maintds}.

\begin{definition}
A \textit{symmetric tripartite graph} refers to a tripartite graph $G=(U, V, W, E)$ that satisfies:
\begin{enumerate}
    \item the graph requirements in Definition~\ref{def:holtri};
    \item if we simultaneously replace every vertex $v\in V$ and its two incident edges $(u, v)$ and $(v, w)$, where $u\in U$ and $w\in W$, with a single edge $(u, w)$, retaining all resulting parallel edges, we get a symmetric bipartite graph.
\end{enumerate}
We write $\symhol(=_d\mid M\mid =_d)$ for any $\hol(=_d\mid M\mid =_d)$ on such symmetric tripartite graphs.
\end{definition}

\begin{proof}[Proof of Theorem~\ref{thm:maintds}]
Due to equivalence \eqref{eq:symequiv}, it suffices to reduce from a known \#P-hard problem to $\symplhol^2([0, 1, 1, 1]\mid =_3)$ on simple graphs.

\begin{figure}[t!]
    \centering
    \includegraphics[width=6cm]{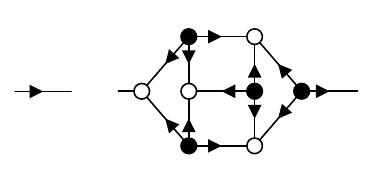}
    \caption{\textit{Left}: The arrow used to represent $M$, highlighting that it is asymmetric and has a sense of direction, with the assignment to the edge connected to the arrow's tail and head corresponding to the row and column index of $M$ respectively. \textit{Right}: The recursive gadget, where solid and hollow circles are $=_3$ on different sides, taking inspiration from \protect\cite{xia2007vcmatching} similar to Figure \ref{pic:ladder}. Note that though we are drawing the recursive gadget with the binary signature $M$, it can also be reinterpreted as being built solely from an arity-3 $f$ and $=_3$, thus being usable in our reductions.}
    \label{pic:minterpol}
\end{figure}

By rewriting the problem as $\symplhol^2(=_3\mid M\mid =_3)$, where $M=\begin{pmatrix}
    -1 & 0\\
    1 & 1
\end{pmatrix}$, we see an opportunity to do interpolation. As depicted in Figure \ref{pic:minterpol}, we define $H_0$ to be a single copy of $M$, and $H_s$ to be the gadget built by starting from $M$ and connecting the recursive gadget $s$ times to the right. We calculate the signature matrix of $H_s$ to be
\begin{equation}
M\cdot\begin{pmatrix}
    4 & 5\\
    6 & 7
\end{pmatrix}^s=MP\begin{pmatrix}
    \lambda_1^s & 0\\
    0 & \lambda_2^s
\end{pmatrix}P^{-1},\label{eq:hs}
\end{equation}
where $\lambda_1=\frac{1}{2}(11-\sqrt{129})$, $\lambda_2=\frac{1}{2}(11+\sqrt{129})$.

Suppose we have a known \#P-hard problem of $\symplhol^2(=_3\mid M'\mid =_3)$ with
\[
M'=MP\begin{pmatrix}
    \alpha & 0\\
    0 & \beta
\end{pmatrix}P^{-1}
\]
for some $\alpha, \beta\in\mathbb{C}$. We consider an arbitrary instance $I$. Let $G$ be the graph of $I$, which has $m$ vertices assigned $M'$. By substituting every $M'$ with $H_s$, the resulting graph $G_s$ is still a symmetric tripartite graph. Similar to \cite[Lemma~2.2]{cai2023hol33}, by formally representing $H_s$ as a chain of binary signatures corresponding to the matrix factors in \eqref{eq:hs} and treating the two diagonal entries as variables, we can write
\[
(\symplhol^2)_{G_s}=\left(\sum_{i+j=m}c_{i, j}(\lambda_1^i\lambda_2^j)^s\right)^2=\sum_{i+j=2m}c_{i, j}'(\lambda_1^i\lambda_2^j)^s.
\]
With Lemma~\ref{lem:powerinterpol}, since $\frac{\lambda_1}{\lambda_2}$ is not a root of unity, we can use the evaluations for $s=1, \dots, 2m+1$ to interpolate all $c_{i, j}'$, and evaluate $\sum_{i+j=2m}c_{i, j}'\alpha^i\beta^j$ to obtain the evaluation of $I$. 

By letting $\alpha=-19-\sqrt{129}, \beta=-19+\sqrt{129}$, we can derive a symmetric $M'$. Specifically:
\begin{gather*}
MP\begin{pmatrix}
    \alpha & 0\\
    0 & \beta
\end{pmatrix}P^{-1}=\begin{pmatrix}
    11 & -5\\
    -5 & -3
\end{pmatrix};\\
(=_3)\begin{pmatrix}
    11 & -5\\
    -5 & -3
\end{pmatrix}^{\otimes 3}=[603, -340, 115, -76].
\end{gather*}
By Lemma~\ref{lem:plhol4} and Lemma~\ref{lem:holtosymhol}, $\symplhol^2([603, -340, 115, -76]\mid =_3)$ is \#P-hard. We have also established a polynomial reduction from this problem to $\symplhol^2([0, 1, 1, 1]\mid\allowbreak =_3)$. Since $s\geq 1$, the constructed graphs are also simple.

To summarize, we have established the following chain of reductions:
\begin{align*}
&\plhol^4([603, -340, 115, -76]\mid =_3) && \text{(\#P-hard; Lemma~\ref{lem:plhol4})}\\
\leq_T{}&\symplhol^2([603, -340, 115, -76]\mid =_3) && \text{(graph-doubling; Lemma~\ref{lem:holtosymhol})}\\
\leq_T{}&\symplshol^2([0, 1, 1, 1]\mid =_3) && \text{(polynomial interpolation)}\\
\equiv_T{}&\symplsgds(\begin{pmatrix}
    [0, 1, 1, 1]\\
    [0, 1, 1, 1]
\end{pmatrix}\mid\begin{pmatrix}
    [0, 1, 1, 1]\\
    [0, 1, 1, 1]
\end{pmatrix}) && \text{(equivalence; Corollary~\ref{cor:symgds})}
\end{align*}
\end{proof}

\section{Conclusion and Discussion}
The $\gds$ framework provides a weighted generalization of \prob{\#$(\sigma, \rho)$-Set}, and serves as our model for studying domination-type counting problems. In Section~\ref{sec:gadget}, we present the theoretical development and application of gadget construction in $\gds$. In Section~\ref{sec:uniform}, we study the connection between $\gds$ problems with specific forms of uniform signatures and the powered versions of $\hol$ problems on usual graph classes as well as symmetric bipartite graphs. These methods establish \#P-completeness for counting dominating sets and total dominating sets on 3-regular planar bipartite simple graphs, and suggest directions for developing more systematic reduction techniques.

For the gadget construction approach, a gadgeture retains the assignment to bridging vertices as input because it affects the evaluation of external vertices, so it cannot automatically be substituted for a vertex signature. Interpolation also depends on specific algebraic properties of the chosen recursive gadgets. Future work could develop new gadgets and interpolation arguments within the present formulation for broader signature families, or explore alternative formulations that allow substitution for vertex signatures or provide other forms of convenience.

For the powered $\hol$ approach, more general hardness results could widen the choice of starting problems. Our hardness proof starts from a constant-sign signature, for which hardness transfers readily from $\hol$ to $\hol^k$. This imposes an additional requirement, since the signatures obtained through available constructions may not always have constant sign even when their ordinary $\hol$ problems are hard.

Another direction concerns the equality between $\gds$ and $\hol^2$. Corollary~\ref{cor:gdstohol2} obtains it for signatures $f_0=(=_d)M^{\otimes d}$ under additional conditions on $M$, without requiring graph symmetry. Alternatively, Corollary~\ref{cor:symgds} gives the equality for any symmetric $f_0$ on symmetric bipartite graphs, where hardness must be established separately. In our proof, Lemma~\ref{lem:holtosymhol} requires the signature used at the graph-doubling step to also satisfy the former conditions; subsequent interpolation leads to the signature for \prob{\#Total Dominating Set}. Extending these routes may require broader conditions for the equality or alternative reductions establishing hardness on symmetric bipartite graphs.

There is potential for establishing complexity dichotomies within the $\gds$ framework. Broader connections with other frameworks of counting problems may also help identify tractable cases and develop hardness reductions beyond those considered here.

\bibliography{ref}

\newpage
\appendix
\renewcommand{\theHsection}{\Alph{section}}
\section{Code for the Ladder Gadget}\label{app:code}
Code \ref{cod:initgad} is a piece of Mathematica code for calculating the initial gadgeture of the ladder gadget. We use \texttt{v1} to \texttt{v4} to refer to the vertices $v_{0, 1}$ to $v_{0, 4}$, and \texttt{v5} to \texttt{v12} to refer to the vertices in $V_I$ in a left-to-right, top-to-bottom order.

\begin{lstlisting}[caption={Mathematica code for calculating the initial gadgeture of the ladder gadget},label=cod:initgad,float=h,abovecaptionskip=-\medskipamount]
f[x_, y_, z_, w_] := 1 - (1 - x) (1 - y) (1 - z) (1 - w)
gds[v1_, v2_, v3_, v4_, v5_, v6_, v7_, v8_, v9_, v10_, v11_, v12_] := 
    f[v1, v3, v4, v7] f[v2, v3, v4, v10] f[v5, v6, v7, v8] f[v5, v6,
    v9, v10] f[v3, v5, v7, v11] f[v5, v8, v9, v11] f[v6, v8, v9, v12] 
    f[v4, v6, v10, v12] f[v7, v8, v11, v12] f[v9, v10, v11, v12]
g0000 = Sum[gds[0, 0, 0, 0, v5, v6, v7, v8, v9, v10, v11, v12],
    {v5, 0, 1}, {v6, 0, 1}, {v7, 0, 1}, {v8, 0, 1}, {v9, 0, 1},
    {v10, 0, 1}, {v11, 0, 1}, {v12, 0, 1}]
\end{lstlisting}

Other values such as \texttt{g0001} can be obtained by changing the assignment to \texttt{v1} through \texttt{v4} accordingly.

\end{document}